\documentclass[12pt]{article}
\usepackage[utf8]{inputenc}
\usepackage[T2A]{fontenc}
\usepackage[english]{babel}
\usepackage{amsthm, amsmath, amssymb, amsfonts}
\usepackage[dvips]{color}
\usepackage{graphicx}
\usepackage{authblk}
\usepackage{url}
\graphicspath{{fig/}}
\allowdisplaybreaks

\newtheorem{thm}{Theorem}[section]

\newtheorem{lemma}[thm]{Lemma}

\newtheorem{theorem}[thm]{Theorem}

\title{\bf Asymptotic Approximation of Fading Mode in Neurooscillator Dynamics \thanks{
The work was carried out in part (formulation of the problem, Sections 4.1 and 4.3 in the proof of the main Theorem 3.3) at the expense of the Russian Science Foundation grant 
No. 22-11-00209, https://rscf.ru/project/22-11-00209/; 
partially (Section 4.2 in the proof of the main Theorem 3.3, Theorem 4.6) within the framework of a development programme for the Regional Scientific and Educational Mathematical Center of the Yaroslavl State University with financial support from the Ministry of Science and Higher Education of the Russian Federation (Аgreement on provision of subsidy from the federal budget No. 075-02-2024-1442).
}}
\author{M. M. Preobrazhenskaia (rita.preo@gmail.com), \\ V. K. Zelenova (verzelenowa12@gmail.com)}
\affil{Centre of Integrable Systems, P. G. Demidov Yaroslavl State University,\\Yaroslavl, Russia}

\begin{document}

\maketitle

\begin{abstract}

We consider a system consisting of two delay differential equations with a large parameter, modeling the association of a pair of neurooscillators. The unknown functions describe the changes in the normalized membrane potentials of neurons over time, with the large parameter characterizing the speed of electrical processes. The first equation is separated from the system and represents a generalized Hutchinson equation. This equation, as known, possesses periodic solutions with high peaks over the period. The second equation is also based on the generalized Hutchinson equation, but with an additional term, linking it to an oscillator satisfying the first equation. For the second equation, it is possible to asymptotically construct the so-called fading neuron mode, which is as follows: for any natural number $n$, one can adjust the parameters of the problem in such a way that the solution is asymptotically close to a periodic function with high peaks over $n$ periods, and then, after a transient process represented by decreasing peaks, becomes asymptotically small.

\end{abstract}

\bigskip

\hspace{.2cm} \textbf{Keywords:}  Delay-differential equation, asymptotics, nonlinear dynamics, generalized Hutchinson equation, fading neurooscillator, large parameter.

\section{Introduction}

Methods of nonlinear dynamics prove to be effective in modelling biological neuron~\cite{Rabinovich_06}. In particular, delay-differential equations play a significant role in this regard  \cite{Kasch_2015,Kasch_1995}. 

The nerve cell possesses a sufficiently complex structure, prompting the consideration of various simplification methods that account for the qualitative properties of the neuron in one way or another. One of the main concepts involves treating the nerve cell as equivalent to an electrical oscillatory circuit with special nonlinear elements (see, for example, \cite{
Kasch_2015,
Kasch_1995,
Hodgkin:Huxley,
GlyKolRoz2013}). 

Phenomenological models are often regarded as simplifications.
In \cite{Kasch_2015, Kasch_1995}, criteria are provided that such models must meet. Among these requirements, the most important is the condition for the presence of a stable periodic impulse regime in the system. One of the models of a single neuron with such a regime is described in \cite{hutch,umn}, and it represents a generalized Hutchinson equation, which forms the basis of the model under consideration.

Among the peculiarities of membrane potential behavior that receive attention when choosing a phenomenological model, the bursting effect stands out \cite{Izhik_2000,Glyzin2013a,ChayRinzel1985,Preob2018}. This phenomenon refers to burst-like activity, where the membrane potential function exhibits several consecutive asymptotically high peaks following a period of asymptotically small values (see Fig. \ref{pic:bur}). 
In this paper, we focus on another effect, which we refer to as the fading oscillator mode (see Fig. \ref{pic:v}). Its characteristic feature is that after several consecutive bursts, fading is observed (a more formal description will be provided below).

\begin{figure}[h]
 \centering
	\includegraphics[width=14cm]{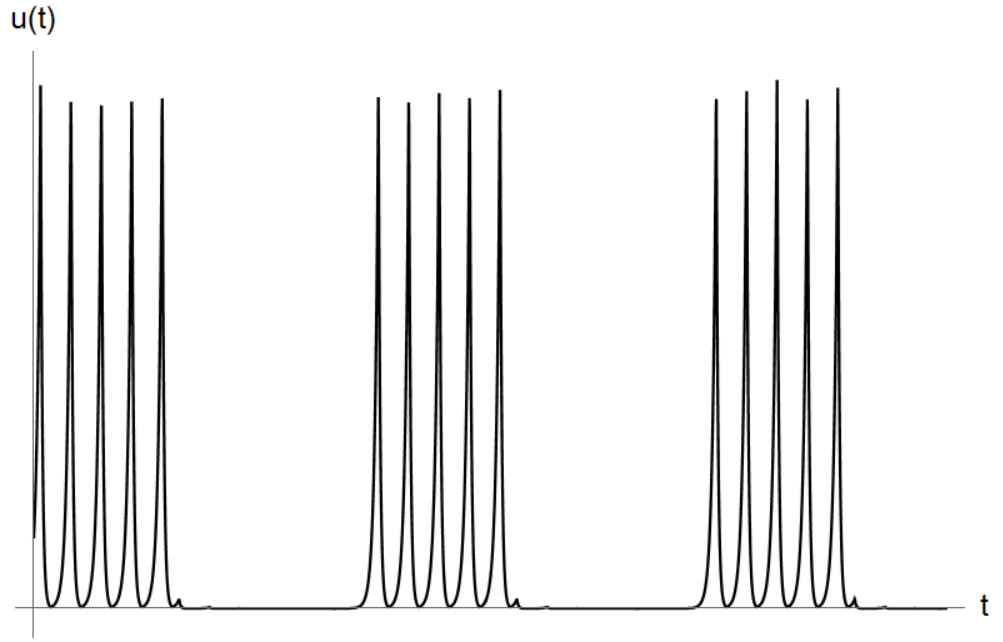}\\[-7pt]
     \centering \caption{Representation of the bursting effect.
    }\label{pic:bur}
\end{figure}

This work considers a system modelling the association of two neuro oscillators. We assume that one of them is isolated and described by a well-known equation \cite{hutch}, ensuring the existence of a periodic regime with one burst per period. At the same time, we presume that the first neuron unilaterally influences the second. In commonly accepted terms (see, for example, \cite{Brown}) in synchronization theory, this means that the first equation is considered the "driver" or "master" while the second is the "response" or "slave". Through this influence, it is possible to achieve a more complex behavior of the second neuron compared to the first. Examining the influence of one neuron on the other helps to better understand how even simple connections between neurons can lead to more complex and diverse behaviors of neural systems. This helps to identify the so-called emergent behavior of the system, which does not manifest when studying only individual elements.

We assume that the first neuron is described by a generalized Hutchinson equation
\begin{equation}
\label{eq_hutch}
    \dot{u}=\lambda f_{\alpha}(u(t-1))u.
\end{equation} 
Here, $u=u(t) > 0$ represents the normalized membrane potential, $\lambda\gg 1$ characterizes the speed of electrical processes, and the smooth function $f_{\alpha}$ satisfies the conditions
\begin{equation}
    \label{cond_f}
    f_{\alpha}(0)=1;\ \lim_{u\rightarrow +\infty} f_{\alpha}(u)=-{\alpha}\ (\alpha>0);\ f_{\alpha}^{'}(u), \ uf_{\alpha}^{''}(u)=O({u^{-2}}),\ u \rightarrow +\infty.
\end{equation}
In the work \cite{hutch}  it was proven that equation (\ref{eq_hutch}) admits a stable relaxation cycle $u^*_{\lambda}(t)=e^{\lambda x^*_{\lambda}(t)}$ of period $T^*_\lambda$, satisfying the conditions:

$$\lim_{\lambda\rightarrow +\infty} \max_t{|x^*_\lambda(t) - x^*(t)|} = 0, \ \lim_{\lambda\rightarrow +\infty} T^*_\lambda=T^*.$$ 
Here, $x^*(t)$  is a periodic piecewise-linear function defined as follows:
   $$ x^*(t)=
    \begin{cases}
        t,\ t\in [0,1],\\
        -\alpha(t-t^*),\ t\in [1, t^*+1],\\
        t-T^*,\ t\in[t^*+1,T^*],
    \end{cases}$$
$$x^*(t+T^*)=x^*(t),\quad t^*=(\alpha+1)/\alpha,\quad T^*=(\alpha+1)^2/\alpha.$$
The functions  $x^*(t)$ and $u^*_{\lambda}(t)$ are depicted in Figure \ref{pic:xu}.

 \begin{figure}[h]
 \centering
	\includegraphics[width=16cm]{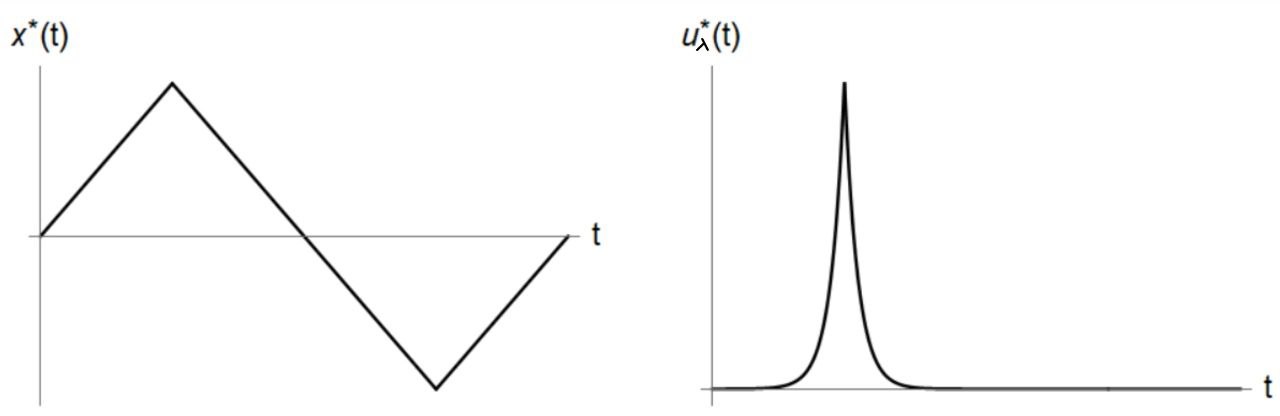}\\[-7pt]
     \centering \caption{Representation of functions $x^*(t)$ and $u^*_\lambda(t)$.
    }\label{pic:xu}

\end{figure}

Regarding the second neurooscillator, we assume that it is coupled to the first one, and its normalized membrane potential $v=v(t)$ satisfies equation 
\begin{equation}
    \label{eq_v}
     \dot{v}=\lambda(f_{\beta}(v(t-h))+g(u(t)))v.
\end{equation}
Here, the parameter $\lambda$ and the function $f_{\beta}$ have the same meaning as in equation (\ref{eq_hutch}); for the function $g$, as well as for $f_{\beta}$, the values are known at $u=0$ and as $u\to+\infty$:
\begin{equation}
\label{cond_g}
    g(0)=-\eta \ (\eta>0);\ \lim_{u\rightarrow +\infty} g(u)=\xi \ (\xi >0);\ g^{'}(u),  ug^{''}(u)=O({u^{-2}}),\ u \rightarrow +\infty.
\end{equation}
As an example of functions $f_\beta$ and $g$ satisfying conditions (\ref{cond_f}) and (\ref{cond_g}), one can provide
\begin{equation}\label{example_fg}
    f_{\alpha}(u)= \dfrac{\alpha(1-u)}{\alpha+ u},\quad g(u)=\dfrac{\xi(u-\eta)}{u+\xi}.
\end{equation}

In the works \cite{GlyPre2019_r,GlyPre2019,preob2022}  a discrete version of equation (\ref{eq_hutch}) was introduced as a phenomenological model of an isolated neuron:
 \begin{equation}
\label{eq_discr_hutch}
 \dot{u}= \lambda \mathcal{F}_{\alpha}(u(t-h))u,
 \end{equation}
where the function
$$\mathcal{F}_{\alpha}(u)=
\begin{cases}
   -\alpha, & u\in(0,1],\\
   1, & u>1,
\end{cases}$$
is chosen to be piecewise-constant with the same values as the function  $f_{\alpha}$ at $u=0$ and $u\to+\infty$.

In the works \cite{perc,pnd_vera} the relay version of equation (\ref{eq_v}) was studied, based on equation (\ref{eq_discr_hutch}):
\begin{equation}
    \label{eq_discr_v}
    \dot{v}=\lambda(\mathcal{F}_{\beta}(v(t-h))+\mathcal{G}(u))v,
\end{equation}
where
    $$\mathcal{G}(u) =  \begin{cases} 
        -\eta & u\in(0,1],\\
        \xi ,&  u>1.
\end{cases}$$

For equation (\ref{eq_discr_v}), the existence of a fading oscillator mode was proven. For any predetermined natural number $n$, one can choose parameters $\alpha$, $\beta$, $\eta$, $\xi$, $h$ such that $v$ coincides with a periodic function $\check{v}(t)$ for the first $n$ periods, and then, after a transient process, coincides with a periodic function $\hat{v}(t)$. Here, the function $\check{v}(t)$ exhibits a high peak (of order $e^\lambda$) over the period, while the function $\hat{v}(t)$ has small values (of order $e^{-\lambda}$). An illustration of such a solution is depicted in Figure \ref{pic:v}.

 \begin{figure}[h]
 \centering
	\includegraphics[width=10cm]{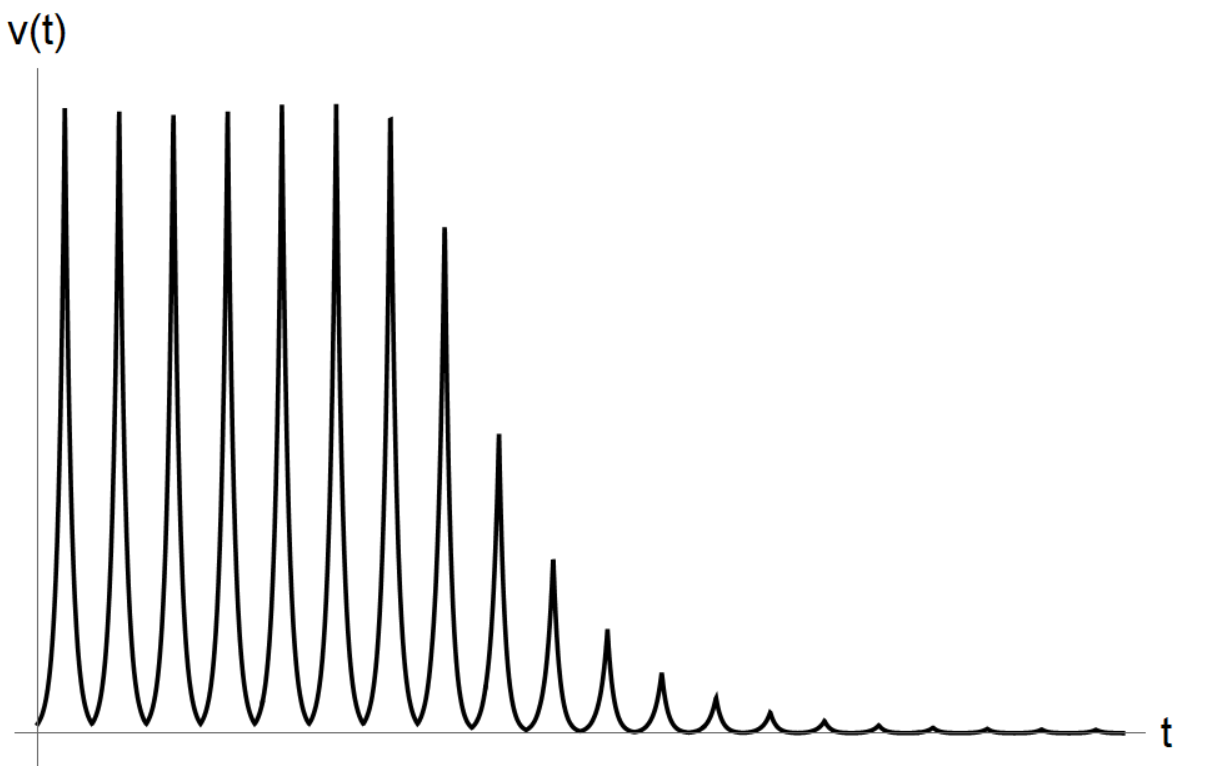}\\[-7pt]
     \centering \caption{Normalized membrane potential $v(t)$ of a fading neurooscillator.
    }\label{pic:v}

\end{figure}

The aim of this work is to analytically prove that equation (\ref{eq_v}), which represents advanced continuous version of the equation, exhibits a fading oscillator mode, namely a solution asymptotically close to the described solution of the relay equation (\ref{eq_discr_v}) over a finite interval, and tending to zero as $t\to+\infty$. This is the focus of Theorem \ref{th_exis}, formulated in Section~3.

\section{Formulation of the problem}

Consider the system of equations
   $$ \begin{cases}
        \dot{u}=\lambda f_{\alpha}(u(t-1))u,\\
        \dot{v}=\lambda(f_{\beta}(v(t-h))+g(u(t)))v.
    \end{cases}$$
The first equation of the system decouples, and as noted above, it has a periodic solution $u^*_{\lambda}(t)=e^{\lambda x^*_{\lambda}(t)}$. In the work \cite{hutch}  it is proven that the function $x^*_{\lambda}(t)$  has the form

\begin{equation}\label{x*_lambda}
   x^*_{\lambda}(t)=
    \begin{cases}
        t + O(e^{-\mu\lambda}),\ t\in [0,1-\sigma_0],\\
        1+\frac{1}{\lambda}u_1(\tau)|_{\tau = \lambda(t-1)} + O(e^{-\mu\lambda}),\ t\in [1-\sigma_0, 1+\sigma_0],\\
        -\alpha(t-t^*) + \frac{d_0}{\lambda}+O(e^{-\mu\lambda}),\ t\in [1+\sigma_0, t^*+1-\sigma_0],\\
        -\alpha+\frac{1}{\lambda}u_2(\tau)|_{\tau = \lambda(t-t^*-1)} + O(e^{-\mu\lambda}),\ t\in [t^*+1-\sigma_0, t^*+1+\sigma_0],\\
        t-T^* + O(e^{-\mu\lambda}),\ t\in[t^*+1+\sigma_0,T^*_{\lambda}],
    \end{cases}
\end{equation}
$$x^*_{\lambda}(t+T^*_{\lambda})=x^*_{\lambda}(t), \  T^*_{\lambda} = T^* + O(e^{-\mu\lambda}),$$
$$u_1(\tau) = \tau + \int\limits_{-\infty}^{\tau}\Big(f_{\alpha}(e^s)-1\Big)ds, \ \ u_2(\tau) = -\alpha\tau + d_0 +  \int\limits_{-\infty}^{\tau}\Big(f_{\alpha}(e^{-\alpha s+d_0})+\alpha\Big)ds,$$
\begin{equation}\label{d_0}
    d_0 = \int\limits_0^1 \dfrac{f_{\alpha}(s)-1}{s}ds + \int\limits_1^{+\infty} \dfrac{f_{\alpha}(s)+\alpha}{s}ds,
\end{equation}
where $\sigma_0$ and $\mu$ are constants satisfying the constraints
\begin{equation}\label{q_sigma_0}
    0<\sigma_0< \min\Big\lbrace\alpha, \frac{1}{2}, \frac{1}{2\alpha}\Big\rbrace, \ \ 0<\mu<\min\Big\lbrace\sigma_0, 1-\sigma_0, 1-\alpha\sigma_0, \alpha-\sigma_0, \alpha-\alpha\sigma_0\Big\rbrace.
\end{equation}

We will seek a function $v$ satisfying equation 
\begin{equation}
\label{eq_v_x*}
\dot{v}=\lambda(f_{\beta}(v(t-h))+g(u^*_{\lambda}(t)))v.
\end{equation}
We will prove that equation  (\ref{eq_v_x*}) has a solution in the form of a fading oscillator, asymptotically close to the solution of equation 
 (\ref{eq_discr_v}) over a finite interval of $t$, and tending to zero as $t\to+\infty$.

For the convenience of subsequent analysis, we will make an exponential substitution $v=e^{\lambda y}$ in equation (\ref{eq_v_x*}) typical for Volterra type equations.
We obtain 
\begin{equation}
    \label{eq_y}
        \dot{y}= f_\beta(e^{\lambda y(t-h)})+g(e^{\lambda x^*_{\lambda}(t)}).
\end{equation}
This equation is the central object of study in the present work.

\section{Main result}

As we take the parameter $\lambda$ in equation (\ref{eq_y}) to $+\infty$, we obtain the limiting relay equation
\begin{equation}
    \label{eq_rele}
        \dot{y}=F(y(t-h))+G(x^*(t)),
\end{equation}
where 
$$F(x) = \lim\limits_{\lambda\to\infty}  f_\beta(e^{\lambda x}) = \begin{cases} 
        1, &  x<0, \\
        - \beta,& x>0,
\end{cases}\quad\quad
G(x) = \lim\limits_{\lambda\to\infty}  g(e^{\lambda x}) = \begin{cases} 
        - \eta,&  x<0, \\
        \xi,&  x>0.
\end{cases}$$

We denote by $l_{k}$ the zeros of the solution of equation (\ref{eq_rele}) in the vicinity of which the function increases, and by $r_{k}$ those in the vicinity of which it decreases. Let us introduce the set of initial functions for equation (\ref{eq_rele}):
\begin{equation}
    \label{set_phi}
    \varphi\in C[-h,0],\ \varphi(0)=-d,\ \varphi(t)<0 \text{ при } t\in[-h,0].
\end{equation}

In the work \cite{pnd_vera}, the following theorem is proven, ensuring the existence of a fading oscillator mode for equation (\ref{eq_y}).

\begin{theorem}
    \label{th_pnd}
        Let us fix a natural number $n$. Let the parameters $\beta, \xi, \eta, h$ satisfy the conditions
     $$ (n-1)T^*+\dfrac{d}{1+\xi} < h < nT^*+\dfrac{d}{1+\xi}, $$
\begin{equation}\label{param}
  \xi > 0, \ \  \eta = 1 + \dfrac{1+\xi}{\alpha},  \ \  \beta > \max\Big\lbrace\dfrac{1+\xi}{n(1+\alpha)}-1, 0\Big\rbrace.
\end{equation}
  Then there exists  $m \in  \mathbb{N} \cup \{0\}$, such that equation  (\ref{eq_rele}) with a continuous initial function from the set  (\ref{set_phi}) under condition
\begin{equation}\label{d}
      0<d<(1+\xi) t^*
  \end{equation}
  has a stable solution
   $$ y^*(t) =
   \begin{cases}
       p^*(t) , & t \in [0, h+l_1],
	\\
	 p^*(t) - (1+\beta)(t - h -l_k) - c_{k-1} , & h+l_k \leq t \leq h +r_k,
	\\
	 p^*(t) -  c_{k} , & h+r_k \leq t \leq h+l_{k+1},
  \\
  p^*(t) - c_{n+m}, & t \geqslant r_{n+m} + h,
   \end{cases}$$
where the function $p^*(t)$ is described by the formula:
\begin{equation}\label{p^*}
    p^*(t) =\begin{cases}
	(1+\xi)t -d , & 0 < t < t^*,
	\\
	 (1-\eta)t + (\xi+\eta)t^* -d, & t^* \leq t \leq T^*,
 \end{cases}  
\end{equation}
$$p^*(t+T^*) = p^*(t),$$
  \begin{equation}\label{c_{k}}
    c_0=0, \ \ c_{k} =  (1+\beta)\sum_{i=1}^{k} (r_i - l_i).
\end{equation} 
 The first $2n$ roots have values:
\begin{equation}
    \label{roots}
    l_{k} =\dfrac{d}{1+\xi} + (k-1)T^*,\ \  r_{k} =-\dfrac{\alpha d}{1+\xi}  + kT^*, \ \ k =  1, \dots, n.
\end{equation} 
\end{theorem}

The constraint (\ref{d}) ensures the existence of roots for the function $p^*(t)$. The choice of parameter $\eta$ guarantees the periodicity of the function $p^*(t)$. With the chosen constraint on the parameter $\beta$ the function $p^*(t)-c_{n+m}$ takes only negative values.

We will prove the statement regarding the relationship between the solutions of equations (\ref{eq_y}) and (\ref{eq_rele}). To this end, we introduce the notation:
$$\Phi(u)=(\eta +\xi)f_{\alpha}(u)+(1+\alpha)g(u).$$
Note that
$$\Phi(0)=\lim\limits_{u\to+\infty}\Phi(u)=-(\alpha+1).$$
We also introduce the notation for the constant:
\begin{equation}\label{d_*}
    d_* = 
\frac{1}{\alpha}\int\limits_0^{+\infty} 
\dfrac{
\Phi(u)+\alpha+1
}{u}du.
\end{equation}

\begin{lemma}\label{sign_d_*}
The parameter $d_*$ can take positive, negative, and zero values.
\end{lemma}
\begin{proof}
Consider the function  $\Phi(u)$ with $f_\alpha(u)$ and $g(u)$, proposed in (\ref{example_fg}), and calculate the value of $d_*$ using formula (\ref{d_*}). We obtain that, in this case, $$d_* = \dfrac{(1+\alpha)^2(1+\xi)}{\alpha^2}\ln{\Big(\dfrac{\alpha}{\xi}\Big)}.$$

Then, $$d_* \begin{cases}
=0, \text{ when } \alpha=\xi,\\
>0, \text{ when } \alpha>\xi,\\
<0, \text{ when } \alpha<\xi.\\
\end{cases}$$

\end{proof}

We fix the parameter $\tilde{d}>0$, which will play the role of the radius of the neighborhood of zero. The following theorem is the main result of the present work.
\begin{theorem}
    \label{th_exis}
Let the parameters $\beta, \xi, \eta, h$ satisfy the constraints of Theorem \ref{th_pnd}. Then, for each $t_0>h$, there exists a sufficiently large $\lambda_0$, such that, for all $\lambda>\lambda_0$, equation (\ref{eq_y}) has a solution $y^*_\lambda(t)$,  satisfying the limit equations
$$\lim\limits_{\lambda\to+\infty}\max_{t\in[0,t_0]}|y_\lambda^*(t)-y^*(t)|=0; 
\quad
\quad
\lim\limits_{t\to+\infty}y_\lambda^*(t)=-\infty \text{ at a fixed }\lambda \text{ and at } d_*<-\tilde{d}.$$
\end{theorem}

This theorem is proved in the following section, and its validity follows directly from Theorems \ref{th_asymp}, \ref{th_asymp_q}, \ref{th_asymp_q2}. The essence of the proof lies in finding the asymptotic formulas for the solution $y^*_\lambda(t)$. For this purpose, we represent the solution as a sum $y^*_\lambda(t)=p^*_\lambda(t)+q^*_\lambda(t)$. It is proven that for any $t_0>0$ on the time interval $[0,t_0]$, the function $p^*_\lambda(t)$
is asymptotically close to the periodic function $p^*(t)$ for sufficiently large $\lambda$; at a fixed $\lambda$ $p^*_\lambda(t)\to-\infty$,  as $t\to+\infty$. Meanwhile, the function $q^*_\lambda(t)$ is asymptotically close to $q^*(t)$, that is, it is small on the interval $[0,h+l_1]$, then on the interval  $[h+l_1, h+r_{n+m}]$ it is close to a decreasing function, and for $t>h+r_{n+m}$ it is asymptotically close to a negative constant (see Fig. \ref{pic:yqp}).

Note that considering the exponential substitution, for the function  $v_\lambda^*(t)=e^{\lambda y_\lambda^*(t)}$ Theorem~\ref{th_exis} implies the existence of a solution depicted in Figure \ref{pic:v}, where  $v_\lambda^*(t)\to 0$, as $t\to+\infty$.

\section{Proof of the theorem \ref{th_exis}}

\subsection{Decomposition into a sum}
We seek the solution of equation (\ref{eq_y})  in the form of a sum $y^*_\lambda(t)=p^*_\lambda(t)+q^*_\lambda(t)$, where $p=p^*_\lambda(t)$  satisfies the equation
\begin{equation}\label{eq_plambda}
    \dot{p}=1+g(e^{\lambda x^*_{\lambda}(t)}),
\end{equation}
and $q=q^*_\lambda(t)$ satisfies the equation
\begin{equation}\label{eq_qlambda}
    \dot{q}= f_\beta(\exp({\lambda (p^*_\lambda(t-h)+q(t-h))}))-1.
\end{equation}

The solution of equation (\ref{eq_rele}) with an initial function from the set (\ref{set_phi}) can also be represented as a sum  $y^*=p^*+q^*$ of the solutions of the following two problems.
The first problem is:
   $$ \begin{cases}
        \dot{p}=1+G(x^*(t)),\\
        p|_{t=0}=-d,
    \end{cases}$$
where the solution $p=p^*(t)$ is a  $T^*$-periodic function described by the formula (\ref{p^*}) (see Fig. \ref{pic:yqp}). 

The second problem is:
\begin{equation}
\label{task_q}
    \begin{cases}
        \dot{q}=
        F(p^*(t-h)+q(t-h))-1,\\
        q(t)=\varphi(t) \text{ при } t\in[-h,0],
    \end{cases}
\end{equation}
where $p^*(t)=-d$ is defined for  $t\in[-h,0]$; $\varphi\in C[-h, 0]$, $\varphi(0)=0$,  and $\varphi(t)<d$,  for $t\in[-h,0]$. The right-hand side of the equation in (\ref{task_q}) is either zero or a negative number $-\beta-1$. Therefore, the solution $q^*(t)$ of problem  (\ref{task_q}) is a continuous function, composed of segments of constants alternating with segments of affine decreasing functions with a slope of $-\beta-1$ (see Fig. \ref{pic:yqp}). 
 \begin{figure}[h]
 \centering
	\includegraphics[width=14cm]{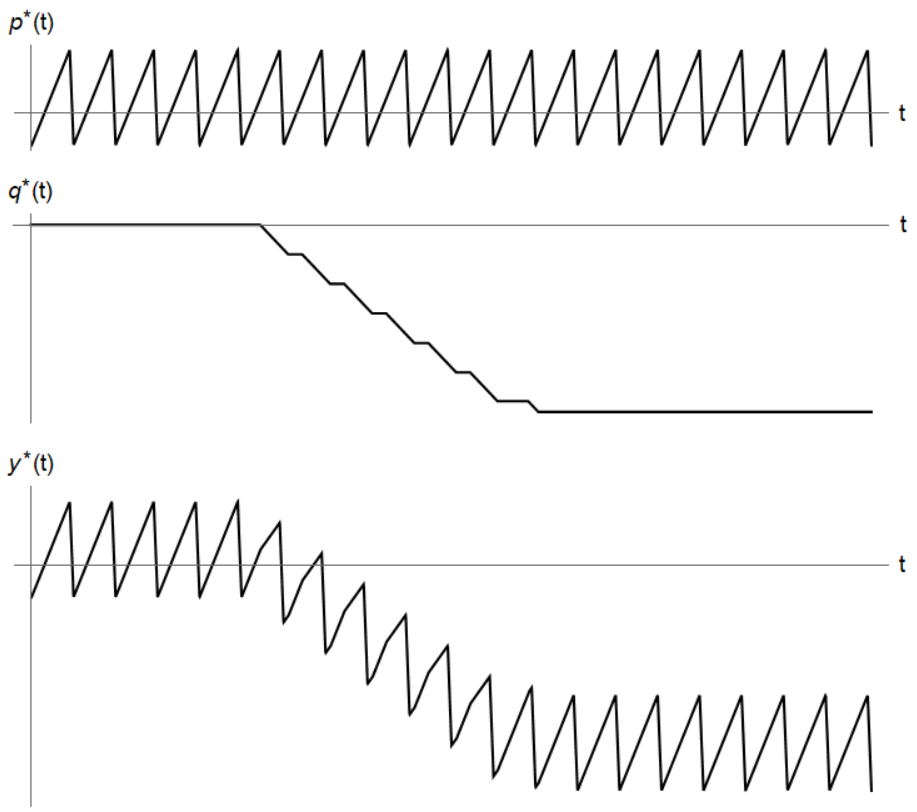}
     \centering \caption{
    The functions $p^*(t)$, $q^*(t)$, $y^*(t)=p^*(t)+q^*(t)$.
    }\label{pic:yqp}
\end{figure}
Moreover, the change in the slope occurs at points of the form $h+l_k$, $h+r_k$, where $l_k$ and $r_k$ are the zeros of the function $y^*(t)=p^*(t)+q^*(t)$.

On the interval  $[0,h+l_1]$ the problem (\ref{task_q}) takes the form $$\dot{q}=0,\quad q|_{t=0}=0,$$ so on this interval $q^*(t)=0$. Consequently, the roots of the equation $p^*(t)+q^*(t)=0$ on this interval are the zeros of the function $p^*(t)$, which are described by formula (\ref{roots}). Next, we inductively find that on the intervals $[h+l_k, h+r_k]$ (which are the intervals of positivity of the already constructed $y^*(t)$ shifted by $h$), $k=1,\ldots,n$, the function $q^*(t)$ is obtained from the problem 
\begin{equation}
    \label{task_q+}
    \dot{q}=-\beta-1, \quad q|_{t=h+l_k}=-c_{k-1},
\end{equation}
where $c_k$ for $k=0,1,\ldots,m+n$ are described by formula (\ref{c_{k}}).
On the intervals $[h+r_k,h+l_{k+1}]$, $k=1,\ldots,n$, the function $q^*(t)$ satisfies
\begin{equation}
    \label{task_q-}
    \dot{q}=0, \quad q|_{t=h+l_k}=-c_k.
\end{equation}
The function $q^*(t)$ does not increase. However, the sum $p^*(t)+q^*(t)$ may have a finite number of roots for $t>r_n$. The number of these roots is denoted by $2m$, where $m\geqslant0$. Consequently, the tasks (\ref{task_q+}) and (\ref{task_q-})  persist on the intervals  $[h+l_k, h+r_k]$, for $k=n+1,\ldots,n+m$, and $[h+r_k,h+l_{k+1}]$, for $k=n+1,\ldots,n+m-1$, respectively.
The inequality $(1+\xi)t^*-d<c_{n}$, which follows from the restriction  (\ref{param}) on the parameter 
 $\beta$,  implies that  $y^*(t)=p^*(t)+q^*(t)<0$ for $t>r_{n}+h$.

Thus,
\begin{equation}
    \label{q*}
    q^*(t)=
    \begin{cases}
        0,& t\in[0,h+l_1],\\
        -(\beta+1)(t-h-l_k)-c_{k-1},&t\in[h+l_k,h+r_k],\ k=1,\ldots,n+m,\\
        -c_k,&t\in[h+r_k,h+l_{k+1}],\ k=1,\ldots,n+m-1,\\
        -c_{m+n},&t\geqslant h+ r_{m+n}.
    \end{cases}
\end{equation}

Next, we will equip equations (\ref{eq_plambda}) and (\ref{eq_qlambda}) with appropriate initial conditions and proceed to separately prove the asymptotic closeness of $p^*_\lambda(t)$ to $p^*(t)$ and $q^*_\lambda(t)$ to $q^*(t)$.

\subsection{The asymptotic behavior of the solution $p^*_\lambda(t)$ to equation  (\ref{eq_plambda})}

Fix the parameter 
\begin{equation}
\label{cond_sigma}
    0<\sigma<\min\left\lbrace\frac{t^*_{\lambda}}{2}, \frac{T^*_{\lambda}-t^*_{\lambda}}{2}\right\rbrace,
\end{equation}
which represents the radius of the neighborhood around the inflection points, i.e., points of the form $kT^*_\lambda,\  kT^*_\lambda+t^*_\lambda$, $k\in \mathbb{N}\cup\{0\}$.

To describe the asymptotic behavior of the function $p^*_\lambda(t)$ in the vicinity of the inflection points, we introduce special functions.
\begin{equation}
    \label{v1}
    v_1(\tau)=(1-\eta)\tau + \int\limits_{-\infty}^{\tau}\Big(g(e^{s})+\eta\Big)ds,
\end{equation}
\begin{equation}
    \label{v2}
    v_2(\tau)= (1+\xi)\tau +\int\limits_{-\infty}^{\tau}\Big(g(e^{-\alpha s+d_0})-\xi\Big)ds,
\end{equation}
Here, the constant  $d_0$ is described by the formula (\ref{d_0}). We define another similar constant.
$$d_1 = \int\limits_{0}^{1}\dfrac{g(u)+\eta}{u}du+ \int\limits_{1}^{+\infty}\dfrac{g(u)-\xi}{u}du. $$
Note that the numbers $d_0$, $d_1$, $d_*$ are related by the equation.
$$d_* = \frac{1}{\alpha}\Big(d_0(\eta +\xi)+ d_1(1+\alpha)  \Big).$$

Let us formulate and prove the statement describing the asymptotic behavior of the introduced functions (\ref{v1}), (\ref{v2}) as $\tau\to-\infty$ and $\tau\to+\infty$.
\begin{lemma}
    \label{lem_v_asymp}
\begin{equation}
    \label{v1_asymp-}
    v_1(\tau)= (1-\eta)\tau +O(e^{\tau}) \text{ as } \tau\to -\infty,
\end{equation}
\begin{equation}
    \label{v1_asymp+}
    v_1(\tau)= (1+\xi)\tau +d_1+O(e^{-\tau}) \text{ as } \tau\to +\infty,
\end{equation}
\begin{equation}
    \label{v2_asymp-}
    v_2(\tau)= (1+\xi)\tau + O(e^{\alpha\tau}) \text{ as } \tau\to -\infty,
\end{equation}
\begin{equation}
    \label{v2_asymp+}
    v_2(\tau)=d_*-d_1+(1-\eta)\tau+O(e^{-\alpha\tau}) \text{ as } \tau\to +\infty,
\end{equation}
\end{lemma}

\begin{proof}

The equality (\ref{v1_asymp-}) follows from the asymptotic representation  $g(u) = -\eta + O(u)$ as $u\rightarrow0$. In order to prove (\ref{v1_asymp+}) we transform the function $v_1(\tau)$ 

$$v_1(\tau) = (1+\xi)\tau+\int\limits_{-\infty}^{0}\Big(g(e^{s})+\eta\Big)ds + \int\limits_{0}^{+\infty}\Big(g(e^{s})-\xi\Big)ds + \int\limits_{\tau}^{+\infty}\Big(g(e^{s})-\xi\Big)ds.$$

We make the substitution $e^s = u$ in the first two integrals, and for the third one, apply the asymptotic equality $g(u) = \xi + O(u^{-1})$ as $u \rightarrow +\infty$. As a result, we obtain the required equality (\ref{v1_asymp+}). 

Formulas (\ref{v2_asymp-}) and (\ref{v2_asymp+}) are proved in a similar manner.

\end{proof}

Let us prove the following theorem concerning the asymptotic behavior of the solution to equation  (\ref{eq_plambda}).

\begin{theorem}\label{th_asymp}
Let the parameters $\beta, \xi, \eta, h$ satisfy the conditions of Theorem \ref{th_pnd}, let $\alpha>0$, the parameters $\sigma_0, \mu$ satisfy the inequalities (\ref{q_sigma_0}), $\lambda\gg1$, $\sigma$ satisfy the inequalities (\ref{cond_sigma}). 
The equation (\ref{eq_plambda}) with the initial value  $p(-\sigma) = -(1-\eta)\sigma-d$ has a solution with the asymptotic behavior:
\begin{equation}
 \label{sol_p*lambda}
 p^*_\lambda(t)= 
\begin{cases}
    -d + \frac{1}{\lambda} v_1(\tau)|_{\tau= \lambda t} + \Delta_1, & t\in[-\sigma,\sigma],\\
  (1+\xi)(t-l_1) + \frac{d_1}{\lambda} + \Delta_2, & t\in[\sigma,t^*_{\lambda}-\sigma],\\
   -d+(1+\xi)t^*_\lambda  + \frac{d_1}{\lambda}+\frac{1}{\lambda}v_2(\tau)|_{\tau = \lambda(t-t^*_\lambda)} +  \Delta_3, & t\in[t^*_{\lambda}-\sigma,t^*_{\lambda}+\sigma],\\
   (1-\eta)(t-r_1)+ \frac{d_*}{\lambda}+\Delta_4, & t\in[t^*_{\lambda}+\sigma,T^*_{\lambda}-\sigma],\\
    - d + \frac{d_*}{\lambda}+ \frac{1}{\lambda}v_1(\tau)|_{\tau = \lambda(t-T^*_\lambda)} + \Delta_5, & t\in[T^*_{\lambda}-\sigma,T^*_{\lambda}+\sigma],\\
    \end{cases}
\end{equation}
\begin{equation}
 \label{sol_p*lambda_per} 
 p^*_\lambda(t+kT_\lambda^*) = p^*_\lambda(t) + k\Big(\dfrac{d_*}{\lambda} + \tilde \Delta\Big),
 \end{equation}
where $\Delta_i = \Delta_i(t,\lambda) = O(\lambda e^{-\mu\lambda})$, $i=1,\ldots,5$,  $\tilde \Delta = \Delta_5(T^*_\lambda) -  \Delta_1(0).$
The asymptotic expressions are written here under the assumption  $\lambda\to\infty$.
\end{theorem}

\begin{proof}
    First, we will prove that on the interval $[-\sigma, T^*_\lambda+\sigma ]$ the function $p^*_\lambda$ described by the formula (\ref{sol_p*lambda}). We will find the asymptotic formulas on each of the intervals $[-\sigma, \sigma], \ [\sigma, t^*_\lambda-\sigma], \ [t^*_\lambda-\sigma, t^*_\lambda+\sigma], \ [t^*_\lambda+\sigma, T^*_\lambda-\sigma]$ and generalize them for all $t>0$.

    The solution to the Cauchy problem
\begin{equation}\label{p_problem}
\begin{cases}
\dot{p} = 1+ g(e^{\lambda x^*_{\lambda}(t)}),\\
p|_{t=-\sigma} = -d - (1-\eta)\sigma
\end{cases}
\end{equation}
has the following form for $ p^*_\lambda(t) = -d - (1-\eta)\sigma + \int\limits_{-\sigma}^{t}\Big(1+g(e^{\lambda x^*_{\lambda}(s)})\Big)ds$.

1. The first interval considered is $[-\sigma, \sigma]$. Here, from equation (\ref{x*_lambda}),  the function $x^*_{\lambda}(t) = t+O(e^{-\mu\lambda})$. 
We will show that the following equality holds: \begin{equation}\label{gO-g=O}
    g(e^{\lambda t + O(e^{-\mu \lambda})}) = g(e^{\lambda t}) + O(e^{-\mu \lambda}).
\end{equation}
The following asymptotic representation holds: $$g(e^{\lambda t + O(e^{-\mu \lambda})})= g(e^{\lambda t})+g'(e^{\lambda t}) \lambda e^{\lambda t}\cdot O( e^{-\mu \lambda}) + O(\lambda e^{-2\mu \lambda}).$$
   Then, due to the property (\ref{cond_g}) of the derivative  $g'$, the quantity  $g'(e^{\lambda t}) \lambda e^{\lambda t}$ is bounded. Hence, the equality (\ref{gO-g=O}) holds true, from which we obtain
\begin{multline*}
     p^*_\lambda(t) = -d - (1-\eta )\sigma+ \int\limits_{-\sigma}^{t}\Big(1+g(e^{\lambda t})  + \lambda O(e^{-\mu\lambda})\Big)ds = \\
    -d + (1-\eta)t + \int\limits_{-\sigma}^{t}\Big(g(e^{\lambda t})+\eta\Big)ds + O(\lambda e^{-\mu\lambda}).
\end{multline*}

Making the substitution $\tau = \lambda t$  in the integral, we obtain the formula for the function $p^*_\lambda(t)$ when $t\in[-\sigma, \sigma]$

$$p^*_\lambda(t) = -d + \frac{1}{\lambda} v_1(\tau)|_{\tau = \lambda t}+ O(\lambda e^{-\mu\lambda}).$$

2.The next interval under consideration is $[\sigma, t^*_\lambda-\sigma]$, where the solution takes the form 

$$ p^*_\lambda(t) = -d + \frac{1}{\lambda} v_1(\tau)|_{\tau = \lambda \sigma} + O(\lambda e^{-\mu\lambda})+\int\limits_{\sigma}^{t}\Big(1+g(e^{\lambda x^*_{\lambda}(s)})\Big)ds.$$ 

The function $x^*_\lambda(t)> \min\{\sigma_0 + O(\frac{1}{\lambda}), \alpha\sigma_0 + O(\frac{1}{\lambda})\} >0 $ on this interval, hence $e^{\lambda x^*_\lambda(t)} \rightarrow +\infty$. Then, utilizing the asymptotic representation  (\ref{v1_asymp+}), we obtain

$$ p^*_\lambda(t) = - d + \frac{1}{\lambda}(1+\xi)\lambda\sigma + \frac{1}{\lambda}(d_1+O(e^{-\lambda\sigma})) + O(\lambda e^{-\mu\lambda})+\int\limits_{\sigma}^{t}\Big(1+ \xi + O(e^{-\lambda \mu}))\Big)ds,$$ 
from which
$$ p^*_\lambda(t) = -d +(1+\xi)t + \frac{d_1}{\lambda} + O(\lambda e^{-\mu\lambda}) \text{ as } t\in[\sigma, t^*_\lambda-\sigma].$$ 

3. Next, we consider the interval   $[t^*_\lambda-\sigma,t^*_\lambda+\sigma ]$. 
Here, the function $x^*_{\lambda}(t)$ changes sign at the point $t^*_{\lambda} = t^* + O(e^{-\lambda \mu})$ and is given by $x^*_{\lambda}(t) = -\alpha(t-t^*) + \frac{d_0}{\lambda} + O(e^{-\lambda \mu}) = -\alpha(t-t^*_\lambda) + \frac{d_0}{\lambda} + O(e^{-\lambda \mu})$. Then, the solution takes the form

$$p^*_\lambda(t) = -d+  (1+\xi)(t^*_\lambda-\sigma) + \frac{d_1}{\lambda} + O(\lambda e^{-\mu\lambda}) + \int\limits_{t^*_\lambda-\sigma}^{t}\Big(1+g(e^{\lambda (-\alpha(s-t^*_\lambda) + \frac{d_0}{\lambda} + O(e^{-\lambda \mu}))})\Big)ds.$$
Denote the integral from the resulting sum as $I_1$. We apply to it an equality analogous to (\ref{gO-g=O}), and make the substitution $\tau = \lambda(t-t^*_{\lambda})$, then

\begin{multline*}
    I_1 =  \frac{1}{\lambda}\int\limits_{-\lambda\sigma}^{\tau}\Big(1+g(e^{-\alpha s+d_0 })\Big)ds + O(\lambda e^{-\lambda \mu}) = \\ 
    \frac{1}{\lambda}(1+\xi)(\tau + \lambda\sigma)+ \frac{1}{\lambda}\int\limits_{-\infty}^{\tau}\Big(g(e^{-\alpha s+d_0})-\xi\Big)ds - \frac{1}{\lambda}\int\limits_{-\infty}^{-\lambda\sigma}\Big(g(e^{-\alpha s+d_0})-\xi\Big)ds  + O(\lambda e^{-\mu\lambda}).
\end{multline*}
We apply the asymptotic equality  $g(u) = \xi + O(u^{-1})$, as $u \rightarrow +\infty$,  and obtain the formula 

$$p^*_\lambda(t) =   -d+ (1+\xi)t^*_\lambda + \frac{d_1}{\lambda} +\frac{1}{\lambda}v_2(\tau)|_{\tau = \lambda(t-t^*_\lambda)} +  O(\lambda e^{-\mu\lambda}),   \text{ as } t\in[t^*_\lambda-\sigma,t^*_\lambda+\sigma ].$$

4. We examine $p^*_\lambda$ on the interval  $[t^*_\lambda+\sigma,T^*_\lambda-\sigma]$. Here, the function $x^*_{\lambda}(t)<\max\{-\alpha\sigma_0 + O(\frac{1}{\lambda}), -\sigma_0 + O(\frac{1}{\lambda})\} < 0.$ Hence, $g(e^{x^*_{\lambda}(t)}) = -\eta + O(e^{-\lambda \mu})$. Then 

\begin{multline*}
     p^*_\lambda(t) =
     - d +(1+\xi)t^*_\lambda + \frac{d_1}{\lambda}+\frac{1}{\lambda}v_2(\lambda\sigma) +  \int\limits_{t^*_\lambda+\sigma}^{t}\Big(1-\eta + O(e^{-\lambda \mu})\Big)ds + O(\lambda e^{-\mu\lambda})=
    \\-d + (\xi+\eta)t^*_\lambda   +(1-\eta)t+ \frac{d_*}{\lambda}+O(\lambda e^{-\mu\lambda}).
\end{multline*}

5. The next interval under consideration is $[T^*_\lambda-\sigma,T^*_\lambda+\sigma]$. Here, $x^*_{\lambda}(t)= t-T^*_{\lambda}$ and changes sign at the point $T^*_{\lambda} = T^* + O(e^{-\mu\lambda})$. Then,

$$p^*_\lambda(t) = -d +(\xi+\eta)t^*_\lambda+(1-\eta)(T^*_\lambda-\sigma) + \frac{d_*}{\lambda}+ \int\limits_{T^*_\lambda-\sigma}^{t}\Big(1+g(e^{\lambda (s-T^*_{\lambda}})\Big)ds+ O(\lambda e^{-\mu\lambda}).$$
Denote the integral from this expression as $I_2$. We make the substitution  $\tau = \lambda(t-T^*_\lambda)$
$$I_2 =\frac{1}{\lambda} \int\limits_{-\lambda\sigma}^{\tau}\Big(1+g(e^{s })\Big)ds  = \frac{1}{\lambda} \int\limits_{-\infty}^{\tau}\Big(g(e^{s})+\eta\Big)ds + \frac{1}{\lambda}(1-\eta)(\tau + \lambda\sigma) +O(\lambda e^{-\lambda \mu}).$$

Then,

$$p^*_\lambda(t) = - d +(1-\eta)T^*_\lambda + (\xi+\eta)t^*_\lambda + \frac{d_*}{\lambda}+ \frac{1}{\lambda}v_1(\tau)|_{\tau = \lambda(t-T^*_\lambda)} + O(\lambda e^{-\mu\lambda}).$$ 
Substituting $T^*_\lambda = T^* + O(e^{-\mu\lambda}), \ t^*_\lambda = t^* + O(e^{-\mu\lambda})$, and utilizing relations  (\ref{param}) for the parameter $\eta$ we obtain

$$ p^*_\lambda(t)=- d + \frac{d_*}{\lambda}+ \frac{1}{\lambda}v_1(\tau)|_{\tau = \lambda(t-T^*_\lambda)} + O(\lambda e^{-\mu\lambda}) \ \text{as} \ t\in[T^*_\lambda-\sigma,T^*_\lambda+\sigma]. $$

Using formulas (\ref{roots}) for the roots of the relay function $p^*(t)$ from Theorem \ref{th_pnd}, we obtain formula (\ref{sol_p*lambda}) for the solution of the Cauchy problem (\ref{p_problem}), when $t\in[-\sigma, T^*_\lambda+\sigma]$.

In the formulation of the problem, it was noted that the function $x^*_\lambda(t)$ is periodic with period $T^*_\lambda$, hence the function  $1+g(e^{x^*_\lambda(t)})$ is also $T^*_\lambda$-periodic. Consider the integral over one period of this function

$$\int\limits_{0}^{T^*_{\lambda}}\Big(1+g(e^{x^*_\lambda(s)})\Big)ds = p^*_{\lambda}(T^*_{\lambda}) - p^*_{\lambda}(0)  = \frac{d_*}{\lambda} +  \Delta_5(T^*_\lambda) -  \Delta_1(0),$$
then 
$$ p^*_{\lambda}(t+kT_\lambda^*) = p^*_{\lambda}(t) + k\Big(\dfrac{d_*}{\lambda} + \tilde \Delta\Big).$$

\end{proof}

\subsection{The asymptotic behavior of the solution $q^*_\lambda(t)$ of equation (\ref{eq_qlambda})}\label{subsec_q}

First, let us describe the asymptotics of the solution $q^*_\lambda(t)$ of equation (\ref{eq_qlambda}) on the interval $[- \sigma, r_n + h + \sigma]$. Here, the numbers $l_k + h$, $r_k + h$, $k = 1, \ldots, n$, which are the first $2n$ break points of the function $q^*(t - h)$, are contained. These are the zeros of $y^*(t)$ shifted by $h$, which coincide with the zeros of the function $p^*(t)$ and are described by the formulas (\ref{roots}).

We will find the asymptotic formulas for the zeros of the solution $p^*_\lambda(t)$ of equation (\ref{eq_plambda}). Let us denote them by $l_{\lambda, k}$, $r_{\lambda, k}$.
Assume the zeros are of the form $l_{\lambda, k} = l_1 + kT_\lambda^* + \Delta$, $r_{\lambda, k} = r_1 + kT_\lambda^* + \Delta$, where $\Delta$ is a small remainder to be determined. Then, considering (\ref{sol_p*lambda}), (\ref{sol_p*lambda_per}),
$$
    p^*_{\lambda}(l_1+kT_\lambda^*+\Delta) = 
  (1+\xi)\Delta + \frac{d_1}{\lambda} + \Delta_2(l_1+\Delta)+ k\Big(\dfrac{d_*}{\lambda} + \Delta_5(T^*_\lambda) - \Delta_1(0)\Big)=0.
$$
From here we obtain for finite $k$
$$\Delta= - \frac{d_1+kd_*}{(1+\xi)\lambda} +O(\lambda e^{\mu\lambda}),$$
from which the asymptotic representation follows:
$$ l_{\lambda,k}=l_k- \frac{d_1+kd_*}{(1+\xi)\lambda} +O(\lambda e^{\mu\lambda}).$$
Similarly, we obtain
$$r_{\lambda,k}=r_k- \frac{d_1+kd_*}{(1+\xi)\lambda} +O(\lambda e^{\mu\lambda}).$$
For finite $k$, by choosing $\lambda$ appropriately, we can achieve $l_{\lambda,k}\in(l_k-\sigma,l_k+\sigma)$ and $r_{\lambda,k}\in(r_k-\sigma,r_k+\sigma)$. We will assume that this inclusion holds for all $k=1,\ldots,n$.

Let us specify the constraints on the parameter $\sigma$. Suppose that in addition to (\ref{cond_sigma}) the inequality
\begin{equation}
\label{cond_sigma2}
    0<\sigma<\frac13\min\left\lbrace\frac{d}{1+\xi}, t^*-\frac{d}{1+\xi},\frac{(1+\xi)t^*-d}{\eta-1},T^*-t^*-\frac{(1+\xi)t^*-d}{\eta-1}\right\rbrace.
\end{equation}
This condition ensures that the $\sigma$-neighborhoods of the points $l_k$, $t^*+(k-1)T^*$, $r_k$, $kT^*$ do not intersect for $k=1, \ldots, n$.

To describe the asymptotics of the function $q_\lambda^*(t)$ on the intervals $[h+l_{k}-\sigma, h+l_{k}+\sigma]$, $[h+r_{k}-\sigma, h+r_{k}+\sigma]$, $k=1, \ldots, n$, we introduce the special functions 
\begin{equation}
    \label{w1}
    w_1(\tau)=\int\limits_{-\infty}^{\tau}( f_\beta(e^{(\xi+1)\theta+d_1})-1)d\theta,
\end{equation}
    \begin{equation}
    \label{w2}
    w_2(\tau)=-(\beta+1)\tau+\int\limits_{-\infty}^{\tau}( f_\beta(e^{(1-\eta)\theta+d_*})+\beta)d\theta.
\end{equation}

Let us formulate and prove the statement about the asymptotic behavior of the introduced functions (\ref{w1}), (\ref{w2}) as $\tau\to-\infty$ and $\tau\to+\infty$.
\begin{lemma}
    \label{lem_w_asymp}
\begin{equation}
    \label{w1_asymp-}
    w_1(\tau)=O(e^{(1+\xi)\tau}) \text{ as } \tau\to -\infty,
\end{equation}
\begin{equation}
    \label{w1_asymp+}
    w_1(\tau)=-(\beta+1)\tau+d_3+O(e^{-(\xi+1)\tau}) \text{ as } \tau\to +\infty,
\end{equation}
where $$d_3=\frac{d_2-d_1(\beta+1)}{\xi+1},
\quad 
d_2=\int\limits_{0}^{1}\frac{ f_\beta(u)-1}{u}du+\int\limits_{1}^{+\infty}\frac{ f_\beta(u)+\beta}{u}du,$$
\begin{equation}
    \label{w2_asymp-}
    w_2(\tau)=-(\beta+1)\tau+O(e^{(\eta-1)\tau})\text{ as } \tau\to -\infty,
\end{equation}
\begin{equation}
    \label{w2_asymp+}
    w_2(\tau)=d_4+O(e^{(1-\eta)\tau}) \text{ as } \tau\to +\infty,
\end{equation}
where
$$d_4=\frac{d_2-d_*(\beta+1)}{\eta-1}.$$
\end{lemma}
\begin{proof}
The equality (\ref{w1_asymp-}) follows from the asymptotic representation $ f_\beta(u)=1+O(u)$ as $u\to0$. To prove (\ref{w1_asymp+}), we transform the integral
\begin{multline*}
\int\limits_{-\infty}^{\tau}( f_\beta(e^{(\xi+1)\theta+d_1})-1)d\theta=
-(\beta+1)\tau-\frac{d_1(\beta+1)}{\xi+1}+\\\int\limits_{-\infty}^{-d_1/(\xi+1)}( f_\beta(e^{(\xi+1)\theta+d_1})-1)d\theta+\int\limits_{-d_1/(\xi+1)}^{+\infty}( f_\beta(e^{(\xi+1)\theta+d_1})+\beta)d\theta-
\int\limits_{\tau}^{+\infty}( f_\beta(e^{(\xi+1)\theta+d_1})+\beta)d\theta. 
\end{multline*}
We make the substitution $u=e^{(\xi+1)\theta+d_1}$ in the first two integrals, and for the third one, we apply the asymptotic equality $ f_\beta(u)=-\beta +O(u^{-1})$ as $u\to+\infty$. As a result, we obtain the required equality (\ref{w1_asymp+}).

Formulas (\ref{w2_asymp-}) and (\ref{w2_asymp+}) are proved similarly.

\end{proof}

Let us introduce a set of initial functions for equation (\ref{eq_qlambda}):
\begin{equation}
    \label{init_set}
    \psi\in C[-h-\sigma,-\sigma],\ \psi(-\sigma)=0,\ -\iota\leqslant\psi(t)-d + (\eta-1)\sigma \leqslant-\nu \text{ as } t\in[-h-\sigma,-\sigma].
\end{equation}
Let us prove the following theorem about the asymptotic behavior of the solution to equation (\ref{eq_qlambda}).

\begin{theorem}\label{th_asymp_q}
Let the parameters $\beta, \xi, \eta, h$ satisfy the conditions of Theorem \ref{th_pnd} and let $\lambda\gg1$. Let $\sigma>0$ satisfy the inequalities (\ref{cond_sigma}) and (\ref{cond_sigma2}).
Then, equation (\ref{eq_qlambda}) with an arbitrary initial function $\psi$ from the class (\ref{init_set}) has a solution $q^*_\lambda(t)$ with the following asymptotic behavior
$$
q^*_{\lambda}(t)= 
\begin{cases}
    q^*_{\lambda,1}(t), & t\in[-\sigma,h+r_1+\sigma],\\
    q^*_{\lambda,k}(t), & t\in[h+r_{k-1}+\sigma,h+r_{k}+\sigma],\ k=2,\ldots,n,
    \end{cases}
$$
where
\begin{equation}\label{q1}
    q^*_{\lambda,1}(t)= 
\begin{cases}
    O(e^{-\nu\lambda}), & t\in[-\sigma,h+l_1-\sigma],\\
    \frac{1}{\lambda}w_1(\tau)|_{\tau=(t-h-l_1)\lambda}+O(e^{-\nu\lambda}), & t\in[h+l_1-\sigma,h+l_1+\sigma],\\
    -(\beta+1)(t-h-l_1)+\frac{d_3}{\lambda}+O(e^{-\nu\lambda}), & t\in[h+l_1+\sigma,h+r_1-\sigma],\\
    -c_1+\frac{1}{\lambda}w_2(\tau)|_{\tau=(t-h-r_1)\lambda}+\frac{d_3}{\lambda}+O(e^{-\nu\lambda}), & t\in[h+r_1-\sigma,h+r_1+\sigma],\\
    \end{cases}
\end{equation}
\begin{equation}\label{qk}
q^*_{\lambda,k}(t)= 
\begin{cases}
    -c_{k-1}+\frac{A_{k-1}}{\lambda}+O(e^{-\nu\lambda}), &t\in[h+r_{k-1}+\sigma,h+l_{k}-\sigma],
    \\
    -c_{k-1}+\frac{A_{k-1}}{\lambda}+\frac{1}{\lambda}w_1(\tau)|_{\tau=(t-h-l_{k})\lambda}+ O(e^{-\nu\lambda}), & t\in[h+l_{k}-\sigma,h+l_{k}+\sigma],\\
    -c_{k-1}-(\beta+1)(t-h-l_{k})+\frac{A_{k-1}+d_3}{\lambda}+O(e^{-\nu\lambda}), & t\in[h+l_{k}+\sigma,h+r_{k}-\sigma],\\
    -c_{k}+\frac{1}{\lambda}w_2(\tau)|_{\tau=(t-h-r_{k})\lambda}+\frac{A_{k-1}+d_3}{\lambda}+O(e^{-\nu\lambda}), & t\in[h+r_{k}-\sigma,h+r_{k}+\sigma].
    \end{cases}
\end{equation}
where
$$A_k=k(d_3+d_4)\mbox{ при } k=1,\ldots,n.$$
The asymptotic expressions here are written under the assumption $\lambda\to\infty$.
\end{theorem}
\begin{proof}

First, we prove that on the interval $[-\sigma,h+r_1+\sigma]$, the function $q^*_\lambda(t)$ is described by formula (\ref{q1}). We will sequentially find asymptotic formulas on the intervals $[-\sigma,h+l_1-\sigma]$, $[h+l_1-\sigma,h+l_1+\sigma]$, $[h+l_1+\sigma,h+r_1-\sigma]$, $[h+r_1-\sigma,h+r_1+\sigma]$ using the step method, relying on the solution $q^*(t)$ of the relay equation (\ref{eq_rele}). Then, by induction, we will prove formula (\ref{qk}).

1. The first considered interval is $[-\sigma,h+l_1-\sigma]$. On the interval $[-\sigma,h-\sigma]$, we have $q(t-h)=\psi(t-h)$, then taking into account the properties (\ref{init_set}) of the function $\psi$, the argument of the function $ f_\beta$ in (\ref{eq_qlambda}) is exponentially small:
$$\exp(\lambda(p^*_\lambda(t-h)+\psi(t-h)))=O(e^{-\nu \lambda}).$$

Using the asymptotic equality $ f_\beta(\varepsilon)=1+O(\varepsilon)$ as $\varepsilon\to0$, we obtain that on the current interval, $q^*_\lambda(t)$ satisfies the Cauchy problem $\dot{q}=O(e^{-\nu\lambda})$, $q|_{t=-\sigma}=0$, from which $q^*_\lambda(t)=O(e^{-\nu\lambda})$. This asymptotic form will be preserved as long as $p^*_\lambda(t-h)+q^*_\lambda(t-h)$ is strictly negative. Thus, we guarantee that
\begin{equation}
    \label{sol_q1}
    q^*_\lambda(t)=O(e^{-\nu\lambda})\text{ as }t\in[-\sigma,h+l_1-\sigma].
\end{equation}

2. The next interval of consideration is $[h+l_1-\sigma,h+l_1+\sigma]$. At the point $l_1$ the function (\ref{q*}) undergoes a break. We assume that the solution has the form
\begin{equation}
    \label{x*lambda_2}
    q^*_\lambda(t)=\frac{1}{\lambda}w_1(\tau)|_{\tau=(t-h-l_1)\lambda}+\Delta,
\end{equation}
where $\Delta$ is the remainder, the smallness of which needs to be proved. 

Substitute (\ref{sol_q1}) and (\ref{x*lambda_2}) into equation (\ref{eq_qlambda}).
$$\frac{{\rm d} w_1}{{\rm d}  \tau}\Big|_{\tau=(t-h-l_1)\lambda}+\dot{\Delta}=
 f_\beta\Big(\exp\big({\lambda \big(p^*_\lambda(t-h)+ O(e^{-\nu\lambda})\big)}\big)\Big)-1.$$
The interval under consideration $[h+l_1-\sigma,h+l_1+\sigma]\subset[h+\sigma,h+t^*_{\lambda}-\sigma]$, therefore here
$$p^*_\lambda(t-h)=(1+\xi)(t-h-l_1) + \frac{d_1}{\lambda} +  O(\lambda e^{-\mu\lambda}).$$ 
Since we have control over $\nu$ we will choose $\nu>0$ such that $\lambda e^{-\mu\lambda}=O(e^{-\nu\lambda})$.
Then, with (\ref{w1})
\begin{equation}\label{eq_q2}
    \dot{\Delta}= f_\beta\Big(\exp\big(\lambda\big((1+\xi)(t-l_1-h)+\frac{d_1}{\lambda}+O(e^{-\nu\lambda})\big)\big)\Big)- f_\beta\Big(\exp\big(\lambda\big((1+\xi)(t-l_1-h)+\frac{d_1}{\lambda}\big)\big)\Big).
\end{equation}
Using the estimate 
$$| f_\beta(x_1)- f_\beta(x_2)|\leqslant\frac{M}{1+\min\{x_1^2,x_2^2\}}|x_1-x_2|\quad \forall x_1,x_2\in\mathbb{R}_+,$$
where
$$M=\sup_{x\in\mathbb{R}_+}(1+x^2)| f_\beta'(x)|<\infty,$$
which is proven in the work  \cite{hutch} for functions of the form (\ref{cond_f}),
we obtain the exponential smallness of the right-hand side of equation (\ref{eq_q2}). 
Then, taking into account the asymptotic behavior (\ref{w1_asymp-}) as $\tau\to-\infty$ of the function $w_1$ when calculating the initial value $\Delta|_{t=h+l_1-\sigma}$, we obtain the Cauchy problem to determine $\Delta$:
    $$\dot{\Delta}=O(e^{-\nu\lambda}),\quad
    \Delta|_{t=h+l_1-\sigma}=O(e^{-\nu\lambda}).$$
Thus, $\Delta=O(e^{-\nu\lambda})$, therefore
\begin{equation}
    \label{sol_q2}
    q^*_\lambda(t)=\frac{1}{\lambda}w_1(\tau)|_{\tau=(t-h-l_1)\lambda}+O(e^{-\nu\lambda})\text{ при }t\in[h+l_1-\sigma,h+l_1+\sigma].
\end{equation}

3. Next, consider the interval  $[h+l_1+\sigma,h+r_1-\sigma]$. Here,
$$q(t-h)=q^*_\lambda(t-h)=O(e^{-\nu\lambda}),$$
\begin{multline*}
p^*_\lambda(t-h)\geqslant\min\{p^*_\lambda(l_1+\sigma),p^*_\lambda(r_1-\sigma)\}=
\\
\min\{(1+\xi)\sigma + \frac{d_1}{\lambda} +  O(\lambda e^{-\mu\lambda}),((\eta-1)\sigma+ \frac{d_*}{\lambda}+O(\lambda e^{-\mu\lambda})\}.
\end{multline*}
Then, on the interval being considered $p^*_\lambda(t-h)+q^*_\lambda(t-h)>
\frac12\min\{(1+\xi)\sigma ,(\eta-1)\sigma\}$.
We choose $\nu>0$ such that $\frac12\min\{(1+\xi)\sigma ,(\eta-1)\sigma\}\geqslant\nu$, then here $p^*_\lambda(t-h)+q^*_\lambda(t-h)>\nu$. Therefore, the argument of the function $ f_\beta$ in (\ref{eq_qlambda})  is exponentially large and has an order of $O(e^{\nu\lambda})$. Considering the asymptotic equality $ f_\beta(u)=-\beta+O(u^{-1})$ as $u\to+\infty$, we obtain that at the current step $q$ is determined by the equation
\begin{equation}
    \label{eq_q3}
    \dot{q}=-\beta-1+O(e^{-\nu\lambda}).
\end{equation}
Substituting the value $t=h+l_1+\sigma$ into (\ref{sol_q2}) and utilizing the asymptotic properties (\ref{w1_asymp+}) of the function $w_1$ as $\tau\to+\infty$, we obtain the initial value
\begin{equation}
    \label{init_q3}
    q|_{t=h+l_1+\sigma}=-(\beta+1)\sigma+\frac{d_3}{\lambda}+O(e^{-\nu\lambda}).
\end{equation}
By solving the initial value problem (\ref{eq_q3}), (\ref{init_q3}), we find that
\begin{equation}
    \label{sol_q3}
    q^*_\lambda(t)=-(\beta+1)(t-h-l_1)+\frac{d_3}{\lambda}+O(e^{-\nu\lambda}) \text{ при } t\in[h+l_1+\sigma,h+r_1-\sigma].
\end{equation}

4. Construct $q^*_\lambda(t)$ on the interval $[h+r_1-\sigma,h+r_1+\sigma]$. We seek the solution here in the form
\begin{equation}
    \label{x*lambda_4}
    q^*_\lambda(t)=-c_1+\frac{1}{\lambda}w_2(\tau)|_{\tau=(t-h-r_1)\lambda}+\frac{d_3}{\lambda}+\Delta,
\end{equation}
where $\Delta$ is the next remainder term, the smallness of which needs to be proven.

Substitute (\ref{x*lambda_4}) and (\ref{sol_q1}) into (\ref{eq_qlambda}):
$$\frac{{\rm d} w_2}{{\rm d}  \tau}\Big|_{\tau=(t-h-r_1)\lambda}+\dot{\Delta}=
 f_\beta\Big(\exp\big({\lambda \big(p^*_\lambda(t-h)+ O(e^{-\nu\lambda})\big)}\big)\Big)-1.$$
Since $[h+r_1-\sigma,h+r_1+\sigma]\subset[h+t^*_{\lambda}+\sigma,h+T^*_{\lambda}-\sigma]$, we have
$$p^*_\lambda(t-h)= (1-\eta)(t-h-r_1)+ \frac{d_*}{\lambda}+O(\lambda e^{-\mu\lambda}).$$ 
From here, taking into account  (\ref{w2}) and the choice of $\nu$, we obtain
\begin{equation}\label{eq_q4}
    \dot{\Delta}= f_\beta\Big(\exp\big(\lambda\big((1-\eta)(t-h-r_1)+ \frac{d_*}{\lambda}+O(e^{-\nu\lambda})\big)\big)\Big)- f_\beta\Big(\exp\big(\lambda\big((1-\eta)(t-h-r_1)+ \frac{d_*}{\lambda}\big)\big)\Big).
\end{equation}
To determine the initial condition, equate the values of the expressions  (\ref{sol_q3}) and (\ref{x*lambda_4}) as $t=h+r_1-\sigma$. Then, taking into account that $c_1=(\beta+1)(r_1-l_1)$ and the asymptotic equality (\ref{w2_asymp-}) for  $w_2(\tau)$ as $\tau\to -\infty$, we obtain
$$\Delta|_{t=h+r_1-\sigma}=O(e^{-\nu\lambda}).$$
The right-hand side of equation (\ref{eq_q4}) is evaluated similarly to how it was done in item 2 of this proof. We obtain that $\Delta$ satisfies the initial Cauchy problem $\dot{\Delta}=O(e^{-\nu\lambda})$, $\Delta|_{t=h+r_1-\sigma}=O(e^{-\mu\lambda})$, hence $\Delta=O(e^{-\mu\lambda})$. Thus,
   $$q^*_\lambda(t)=-c_1+\frac{1}{\lambda}w_2(\tau)|_{\tau=(t-h-r_1)\lambda}+\frac{d_3}{\lambda}+O(e^{-\nu\lambda})\text{ for }t\in[h+r_1-\sigma,h+r_1+\sigma].$$

Next, we perform an induction step. Note that due to (\ref{sol_q1}), for $t\in[h-\sigma,2h+l_1-\sigma]$, the representation $q^*_\lambda(t)=O(e^{-\nu\lambda})$ holds. This means that the roots of the function $y^*_\lambda(t)$ differ from the roots of $p^*_\lambda(t)$ by an amount of order $O(e^{-\nu\lambda})$. Suppose that on the next interval, the function $q^*_\lambda(t)$ coincides with $q^*_{\lambda,k}(t)$, and we will prove that $q^*_{\lambda}(t)=q^*_{\lambda,k+1}(t)$ for $t\in[h+r_{k-1}+\sigma,h+r_{k}+\sigma]$, applying constructions similar to those given above.

\end{proof}

To describe the function $q$ in the vicinity of the break points $h+l_k$, $h+r_k$, $k=n+1,\ldots,n+m$, we will need functions similar to (\ref{w1}), (\ref{w2}), but described in a more general form, since the corresponding roots $l_k$, $r_k$, $k=n+1,\ldots,n+m$ of the function $y^*(t)$ are not described by exact formulas to avoid case considerations. 
We introduce the following functions:
\begin{equation}
    \label{W1}
    W_{l,k}(\tau)=\int\limits_{-\infty}^{\tau}\Big( f_\beta\big(\exp\big({\dot{y}^*(l_{k})\theta+L_{k}}\big)\big)-1\Big)d\theta,
\end{equation}
    \begin{equation}
    \label{W2}
    W_{r,k}(\tau)=-(\beta+1)\tau+\int\limits_{-\infty}^{\tau}\Big( f_\beta\big(\exp\big({\dot{y}^*(r_{k})\theta+R_{k}}\big)\big)+\beta\Big)d\theta.
\end{equation}
Here,
\begin{equation}
    \label{Bk}
    L_{k}=\lim\limits_{\lambda\to+\infty}\lambda(y^*_\lambda(t)-y^*(t)) \text{ as } t\in[l_{k}-\sigma,l_{k}+\sigma],
\end{equation}
\begin{equation}
    \label{Ek}
    R_{k}=\lim\limits_{\lambda\to+\infty}\lambda(y^*_\lambda(t)-y^*(t)) \text{ as } t\in[r_{k}-\sigma,r_{k}+\sigma].
\end{equation}
The limits (\ref{Bk}) and (\ref{Ek}) exist and have finite values due to the form of the constructed function  $y^*_\lambda(t)=p^*_\lambda(t)+q^*_\lambda(t)$.

For the functions (\ref{W1}), (\ref{W2}) the following statement holds.
\begin{lemma}
    \label{lem_W_asymp}
\begin{equation}
    \label{W1_asymp-}
    W_{l,k}(\tau)=O(e^{\dot{y}^*(l_{k})\tau})\text{ as }\tau\to-\infty,
\end{equation}
\begin{equation}
    \label{W1_asymp+}
    W_{l,k}(\tau)=-(\beta+1)\tau+D_{3,k}+O(e^{-\dot{y}^*(l_{k})\tau}) \text{ as } \tau\to +\infty,
\end{equation}
where $$D_{3,k}=\frac{d_2-(\beta+1)L_{k}}{\dot{y}^*(l_{k})},$$
\begin{equation}
    \label{W2_asymp-}
    W_{r,k}(\tau)=-(\beta+1)\tau+O(e^{-\dot{y}^*(r_{k})\tau}) \text{ as } \tau\to -\infty,
\end{equation}
\begin{equation}
    \label{W2_asymp+}
    W_{r,k}(\tau)=D_{4,k}+O(e^{\dot{y}^*(r_{k})\tau}) \text{ as } \tau\to +\infty,
\end{equation}
where
$$D_{4,k}=\frac{d_2-(\beta+1)R_k}{\dot{y}^*(r_{k})}.$$
\end{lemma}
Note that, for $k=1\ldots,n$,  the functions $W_{l,k}(\tau)$, $W_{r,k}(\tau)$ and the constants $D_{3,k}$, $D_{4,k}$ coincide with the functions $w_1(\tau)$ $w_2(\tau)$ and the constants $d_3$, $d_4$, respectively.

The following fact is an extension of Theorem \ref{th_asymp_q}.

\begin{theorem}\label{th_asymp_q2}
Let the conditions of Theorem \ref{th_asymp_q} hold. Suppose that $d_*<0$.  Then, the following holds:
$$
q^*_\lambda(t)= 
\begin{cases}
    q^*_{\lambda,k}(t), &t\in[h+r_{k-1}+\sigma,h+r_{k}+\sigma],\ k=n+1,\ldots,n+m,\\
   -c_{n+m}+\frac{A_{n+m}}{\lambda}+O(e^{-\nu\lambda}), &t\in[h+r_{n+m}+\sigma,+\infty),
    \end{cases}
$$
$$
q^*_{\lambda,k}(t)= 
\begin{cases}
    -c_{k-1}+\frac{A_{k-1}}{\lambda}+O(e^{-\nu\lambda}), &t\in[h+r_{k-1}+\sigma,h+l_{k}-\sigma],\\
    -c_{k-1}+\frac{A_{k-1}}{\lambda}+\frac{1}{\lambda}W_{l,k}(\tau)|_{\tau=(t-h-l_{k})\lambda}+ O(e^{-\nu\lambda}), & t\in[h+l_{k}-\sigma,h+l_{k}+\sigma],\\
    -c_{k-1}-(\beta+1)(t-h-l_{k})+
    \frac{A_{k-1}+D_{3,k}}{\lambda}+O(e^{-\nu\lambda}), & t\in[h+l_{k}+\sigma,h+r_{k}-\sigma],\\
   -c_{k}+\frac{1}{\lambda}W_{r,k}(\tau)|_{\tau=(t-h-r_{k})\lambda}+
   \frac{A_{k-1}+D_{3,k}}{\lambda}+O(e^{-\nu\lambda}), & t\in[h+r_{k}-\sigma,h+r_{k}+\sigma],\\
   
    \end{cases}
$$
where
$$A_k=n(d_3+d_4)+\sum_{i=n+1}^{k}(D_{3,i}+D_{4,i})\mbox{ при } k=n+1,\ldots,n+m.$$
\end{theorem}
\begin{proof}
   The coincidence of $q^*_\lambda(t)$ with $q^*_{\lambda,k}(t)$ on the intervals $[h+r_{k-1}+\sigma,h+r_{k}+\sigma]$, $k=n+1,\ldots,n+m$, is proven using arguments similar to those presented in the proof of Theorem \ref{th_asymp_q}. Note that the constraint on the parameter $\sigma$ needs to be specified such that the $\sigma$-neighborhoods of the inflection points and zeros of the function $y_\lambda^*(t)$ do not overlap.

We will prove that $q^*_\lambda(t)=-c_{n+m}+\frac{A_{n+m}}{\lambda}+O(e^{-\nu\lambda})$ for $t\in[h+r_{n+m}+\sigma,+\infty)$. Consider the interval $[h+r_{n+m}+\sigma,2h+r_{n+m}+\sigma]$. Here, the function $q^*_\lambda(t)$ is found from the equation $\dot{q}= f_\beta(e^{\lambda y^*_\lambda(t-h)})-1$, where $y^*_\lambda(t-h)=p^*_\lambda(t-h)+q^*_\lambda(t-h)<-\nu$ (the validity of this inequality can be ensured by the choice of $\nu$). Hence, here 
    \begin{equation}
        \label{eq_qlast}
        \dot{q}=O(e^{-\nu\lambda})
    \end{equation}
    as $t\in[h+r_{n+m}+\sigma,2h+r_{n+m}+\sigma]$. We compute the initial value
    \begin{multline*}
        q|_{t=h+r_{n+m}+\sigma}=q_{\lambda,n+m}(h+r_{n+m}+\sigma)=\\\Big(-c_{n+m}+\frac{1}{\lambda}W_{r,n+m}(\tau)|_{\tau=(t-h-r_{n+m})\lambda}+
   \frac{A_{n+m-1}+D_{4,n+m}}{\lambda}+O(e^{-\nu\lambda})\Big)\Big|_{t=h+r_{n+m}+\sigma}.
    \end{multline*}
   Taking into account the asymptotic properties (\ref{W2_asymp+}), we obtain
   \begin{equation}
        \label{init_qlast}
        q|_{t=h+r_{n+m}+\sigma}=-c_{n+m}+
   \frac{A_{n+m}}{\lambda}+O(e^{-\nu\lambda}).
   \end{equation}
   From the initial Cauchy problem (\ref{eq_qlast}), (\ref{init_qlast}), we obtain 
   \begin{equation}
       \label{sol_qlast}
       q^*_\lambda(t)=-c_{n+m}+
   \frac{A_{n+m}}{\lambda}+O(e^{-\nu\lambda})\text{ as }t\in[h+r_{n+m}+\sigma,2h+r_{n+m}+\sigma].
   \end{equation}
   We will show that this formula holds for $t\in[2h+r_{n+m}+\sigma,+\infty)$. We perform the induction step over the number of steps of length $h$.
Suppose that $q^*_\lambda(t)=-c_{n+m}+
   \frac{A_{n+m}}{\lambda}+O(e^{-\nu\lambda})$ 
   for $t\in[sh+r_{n+m}+\sigma,(s+1)h+r_{n+m}+\sigma]$, We will prove that this is true on the interval $t\in[(s+1)h+r_{n+m}+\sigma,(s+2)h+r_{n+m}+\sigma]$, $s\in\mathbb{N}$. For each $t$ here, there exists a $k$ such that $t-kT_\lambda^*\in[0,T_\lambda^*]$. Then, taking into account (\ref{sol_p*lambda_per}), we obtain
   \begin{equation}
       \label{y*k}
       y_\lambda^*(t)=p^*_{\lambda}(t)+q^*_{\lambda}(t) = p^*_{\lambda}(t-kT_\lambda^*) + k\Big(\dfrac{d_*}{\lambda} + \tilde \Delta\Big)+q^*_{\lambda}(t).
   \end{equation}
   Since $\max\limits_{t\in[0,T_\lambda^*]}p^*_{\lambda}(t)=(1+\xi)t^*-d+O(1/\lambda)$, from the condition (\ref{param}), it follows that for sufficiently large $\lambda$
   $ p^*_{\lambda}(t-kT_\lambda^*)<c_n$. Then, due to the negativity of $d_*$ and (\ref{sol_qlast}), we obtain
   $$y_\lambda^*(t)<c_n -c_{n+m}+
   \frac{A_{n+m}}{\lambda}+O(e^{-\nu\lambda})<-\nu\text{ as }t\in[sh+r_{n+m}+\sigma,(s+1)h+r_{n+m}+\sigma].$$
   This means that
$$y_\lambda^*(t-h)<-\nu\text{ for }t\in[(s+1)h+r_{n+m}+\sigma,(s+2)h+r_{n+m}+\sigma],$$
Then from equation (\ref{eq_qlast}) with the initial condition $q|_{t=(s+1)h+r_{n+m}+\sigma}=-c_{n+m}+
   \frac{A_{n+m}}{\lambda}+O(e^{-\nu\lambda})$
   on the interval $[(s+1)h+r_{n+m}+\sigma,(s+2)h+r_{n+m}+\sigma]$ we get $q^*_\lambda(t)=-c_{n+m}+
   \frac{A_{n+m}}{\lambda}+O(e^{-\nu\lambda})$. Thus, the induction step is completed.

\end{proof}

Therefore, the functions $p^*_\lambda(t)$, $q^*_\lambda(t)$, and consequently $y^*_\lambda(t)$ are constructed. The latter turned out to be asymptotically close to $p_\lambda^*(t)$ for $t\in[-\sigma,h+l_1-\sigma]$, and after a certain transitional process, it becomes negative.
Note that from the equality (\ref{y*k}) and the negativity of $d_*$, it follows that for a fixed $\lambda$ and large $k$, the function $y_\lambda^*(t)$ can take arbitrarily large negative values.

\section{Conclusion}

We proved Theorem \ref{th_exis}, which, due to the exponential substitution $v = e^{\lambda y}$, implies the following. Equation (\ref{eq_v}) has a solution $v_\lambda^*(t) = e^{\lambda y_\lambda^*(t)}$, which is asymptotically close to a periodic solution with high (exponentially large in $\lambda$) spikes over $n$ periods. Then, after a transitional process, it stabilizes as $v_\lambda^*(t) = O\left(\frac{1}{\lambda}\right)$. Thus, for equation (\ref{eq_v}), there exists a regime of a fading oscillator (see Fig.~\ref{pic:v}).

This result is obtained under the assumption $d_* < -\tilde{d}$, which guarantees that the function $v_\lambda^*(t)$ approaches zero as $t \to +\infty$. A continuation of this work could involve constructing a solution under the assumption $d_* > \tilde{d}$. Note that in this case, the term $k\left(\frac{d_*}{\lambda} + \tilde{\Delta}\right)$ in (\ref{sol_p*lambda_per}) is positive, therefore, for a fixed $\lambda$, $p^*_\lambda(t) \to +\infty$. However, as shown in section \ref{subsec_q}, for such $t$ that the sum $p^*_\lambda(t-h) + q^*_\lambda(t-h)$ is positive, the function $q^*_\lambda(t-h)$ becomes asymptotically close to an affine decreasing function (with a slope of $-\beta - 1$), requiring additional investigation to determine the behavior of the function $y^*_\lambda(t)$.

Moreover, the investigation of the periodicity of $p_\lambda^*(t)$ is beyond the scope of this paper. For this, the quantity $k\left(\frac{d_*}{\lambda} + \tilde{\Delta}\right)$ must equal zero. Deriving the conditions under which this is possible could also be a continuation of the present work.

\subsection*{Conflict of interest} The authors have no conflicts of interest.


\begin{thebibliography}{1}

\bibitem{Rabinovich_06} 
Rabinovich, M.I., Varona, P., Selverston, A.I., Abarbanel
H.D.I.:  Dynamical principles in neuroscience. Rev. Mod. Phys. 78, 1213--1265. (2006)
 \url{doi:10.1103/RevModPhys.78.1213}


 \bibitem{Kasch_2015} Kashchenko, S.: Models of Wave Memory. 
  Springer International Publishing, Switzerland (2015). 
 \url{doi:10.1007/978-3-319-19866-8} 

 \bibitem{Kasch_1995} Kashchenko, S.A., Maiorov, V.V., and Myshkin, I.Yu.: 
Wave distribution in simplest ring neural structures. 
 Matem. Mod. 7:12.  3--18 (1995).
\url{http://mi.mathnet.ru/mm1392} 

 \bibitem{Hodgkin:Huxley} Hodgkin, A.L., and Huxley, A.F.:  
A quantitative description of membrane current and its application to conduction and excitation in nerve.  J. Physiol. 117, 500--544 (1952). 

\bibitem{GlyKolRoz2013} Glyzin,  S.D.,  Kolesov,  A.Y.,  Rozov,  N.K.:   
On  a  method  for  mathematical  modeling  of  chemical
synapses. Differential Equations. 49(10), 1193--1210 (2013).
\url{doi: 10.1134/S0012266113100017} 

\bibitem{hutch}
Kolesov А.Y.,  Mishchenko Е.F., Rozov N.K. A modification of Hutchinson’s equation,  Comput. Math. Math. Phys., 50:12, 1990--2002 (2010). \url{doi:10.1134/S0965542510120031}

\bibitem{umn}  Glyzin S.D., Kolesov A.Y., Rozov N.K., Self-excited relaxation oscillations in networks of impulse neurons, Russian Mathematical Surveys, 70(3), 383--452 (2015).
\url{doi:10.1070/RM2015v070n03ABEH004951}

\bibitem{Brown} Brown R., Kocarev L. // Chaos. 2000. V.10. N2. P. 344-349. \url{doi:10.1063/1.1665004}

\bibitem{Izhik_2000}  Izhikevich, E.:  Neural excitability, spiking and
bursting. International Journal of Bifurcation and Chaos. 
10(6), 1171--1266 (2000). 
\url{doi:10.1142/S0218127400000840}


\bibitem{Glyzin2013a} Glyzin, S.D.,  Kolesov,  A.Y.,  Rozov,  N.K.: 
Modeling the bursting effect in neuron systems. Mathematical Notes.
93(5), 676--690 (2013).
\url{doi:10.1134/S0001434613050040}

\bibitem{ChayRinzel1985} 
Chay, T.R., Rinzel, J.: Bursting, beating, and chaos in an excitable membrane model. 
Biophys. J. 47(3), 357--366 (1985).
\url{doi:10.1016/S0006-3495(85)83926-6}

\bibitem{Preob2018} 
Preobrazhenskaia, M.M.:
The Impulse-Refractive Mode in a Neural Network with Ring Synaptic Interaction. Automatic Control and Computer Sciences. 52(7), 777--789 (2018).
\url{doi:10.3103/S0146411618070210}

\bibitem{GlyPre2019_r} Glyzin, S.D., Preobrazhenskaia, M.M.:
Multistability and Bursting in a Pair of Delay Coupled Oscillators With a Relay Nonlinearity. IFAC PAPERSONLINE. 52(18), 109--114 (2019). \url{doi:10.1016/j.ifacol.2019.12.215}

\bibitem{GlyPre2019} Glyzin,  S.D., Preobrazhenskaia,  M.M.: 
Two Delay-Coupled Neurons with a Relay Nonlinearity. In: Kryzhanovsky B. and etc. (eds) Adv. in Neural Comp., Mach. Learning, and Cogn. Research III. NEUROINFORMATICS 2019. Studies in Comp. Int. vol 856, 181--189. Springer, Cham. (2020).
\url{doi: 10.1007/978-3-030-30425-6_21}

\bibitem{preob2022} Preobrazhenskaia, M.M.: Three Unidirectionally Synaptically Coupled Bursting Neurons. In: Kryzhanovsky, B., Dunin-Barkowski, W., Redko, V., Tiumentsev, Y., Klimov, V.V. (eds) Advances in Neural Computation, Machine Learning, and Cognitive Research V. NEUROINFORMATICS 2021. Studies in Computational Intelligence, vol 1008. Springer, Cham. (2022). \url{doi:10.1007/978-3-030-91581-0_18}

\bibitem{perc} Preobrazhenskaia~M.M. (2023). Relay System of Differential Equations with Delay as a Perceptron Model. In: Kryzhanovsky, B., Dunin-Barkowski, W., Redko, V., Tiumentsev, Y. (eds) Advances in Neural Computation, Machine Learning, and Cognitive Research VI. NEUROINFORMATICS 2022. Studies in Computational Intelligence, vol 1064. Springer, Cham. \url{doi:10.1007/978-3-031-19032-2_53}

\bibitem{pnd_vera} Zelenova VK. Relay model of a fading neuron. Izvestiya VUZ. Applied Nonlinear Dynamics. 32(2), 268–284 (2024). 
\url{doi: 10.18500/0869-6632-003096}




\end{thebibliography}
\end{document}